\documentclass[11pt, a4paper]{amsart}

\usepackage[utf8]{inputenc}
\usepackage[T1]{fontenc}
\usepackage[english]{babel}
\usepackage{amsmath}
\usepackage{amsfonts}
\usepackage{ae}
\usepackage{units}
\usepackage{icomma}
\usepackage{color}
\usepackage{graphicx}
\usepackage{amsthm}
\usepackage{amssymb}
\usepackage{cancel}
\usepackage{fullpage}
\usepackage{upgreek}
\usepackage{type1cm}
\usepackage{eso-pic}
\usepackage{enumerate}
\usepackage[breaklinks,pdfpagelabels=false]{hyperref}

\newcommand{\bm}[1]{\mbox{\boldmath $ {#1} $}}

\renewcommand{\epsilon}{\varepsilon}

\newtheorem{theorem}{Theorem}[section]

\newtheorem{proposition}[theorem]{Proposition}

\theoremstyle{definition}
\newtheorem{definition}[theorem]{Definition}
\newtheorem{rem}[theorem]{Remark}

\numberwithin{equation}{section}
\numberwithin{theorem}{section}

\begin{document}

\title[The Hegselmann-Krause dynamics on the circle converge] 
{The Hegselmann-Krause dynamics on the circle converge}

\author{Peter Hegarty$^{1,2}$, Anders Martinsson$^{1,2}$ and Edvin Wedin$^{1,2}$} 
\address{$^1$Mathematical Sciences, Chalmers, 41296 Gothenburg, Sweden} 
\address{$^2$Mathematical Sciences, University of Gothenburg,  41296 Gothenburg, Sweden} 
\email{hegarty@chalmers.se}
\email{andemar@chalmers.se}
\email{edvinw@student.chalmers.se}





\date{\today}

\begin{abstract}
We consider the Hegselmann-Krause dynamics on a one-dimensional torus and provide the first proof of convergence of this system. The proof requires only fairly minor modifications of existing methods for proving convergence in Euclidean space.
\end{abstract}

\maketitle

\section{Introduction}\label{sect:intro}

The so-called Hegselmann-Krause bounded confidence model (HK-model for brevity), introduced in \cite{HK}, is one of the most popular mathematical models for the dynamics of opinion formation in groups of interacting agents. In its classical formulation, we have a finite number $n$ of agents, indexed by the integers $1,\, 2,\dots,\,n$. 
Time is measured discretely and the opinion of agent $i$ at time $t \in \mathbb{N} \cup \{0\}$ is represented by a real number $x_{t}(i)$. There is a fixed parameter $r > 0$ such that the dynamics are given by
\begin{equation}\label{eq:update}
x_{t+1}(i) = \frac{1}{|\mathcal{N}_{t}(i)|} \sum_{j \in \mathcal{N}_{t}(i)} x_{t}(j),
\end{equation}
where $\mathcal{N}_{t}(i) = \{j : |x_{t}(j) - x_{t}(i)| \leq r \}$. In words, at each time step an agent takes account of his so-called \emph{$r$-neighbourhood}, consisting of those other agents whose opinions currently lie within distance $r$ of his own. He updates his opinion to the average of those held by the members of his $r$-neighbourhood, including himself. 
\par In purely formal terms, the model makes sense if opinions $x_{t}(i)$ are assumed to come from a set $V$ which has enough structure so that it is possible to make sense of the command to 
\\
\par \emph{`move to the average of a finite collection of points within distance $r$ of your present location'}. 
\\
\\
We shall refer to any rigorous formulation of this procedure as a \emph{HK-update rule}. The simplest generalisation would be to $\mathbb{R}^k$ for any $k \geq 1$ and with the Euclidean metric. This is natural if we imagine that there are $k$ `issues' on which agents have opinions, and they are only willing to compromise with those whose opinions on \emph{every} issue are sufficiently close to their own.
There are rigorous results in the literature concerning the HK-dynamics in Euclidean space of arbitrary dimension, see for example \cite{BBCN}. Note, though, that it is not clear what is the most natural formulation of the HK-dynamics in $\mathbb{R}^k$ for $k > 1$. For example, the $L^{\infty}$-metric may better capture the notion that an agent will only compromise with those whose opinions on all issues are close to their own -- indeed, the $L^2$-metric seems better suited when considering the HK-dynamics as a rendezvous procedure, as in for example \cite{MBCF}. However, it is also plausible that the HK-update rule is too strict a criterion for compromise. At least one group of authors \cite{ZS} has suggested a more general model in which there is a second parameter $l \leq k$ such that one agent will compromise with another whose opinions on at least $l$ of the $k$ issues are within a fixed distance of his own. Note that, in this generality, we are no longer dealing with a HK-update rule, as formulated above, since if $l < k$ then the `distance' between agents does not satisfy the triangle inequality. 
\par If we want to stick with the HK-update rule, then perhaps the simplest and most natural example to consider after Euclidean space is a circle. Let $\mathcal{C}_p$ denote a circle of perimeter $p$, equipped with its natural metric which we will denote by $\delta(\cdot, \, \cdot)$. If $p \leq 2r$ then the `average location of those in your $r$-neighbourhood' is not well-defined, unless there is a priori agreement on an orientation of the circle and an origin, i.e.: on a continuous bijection from $\mathcal{C}_p$ to $[0,\,p)$. In any case, if such a priori agreement existed, then the dynamics would become trivial as any configuration would collapse immediately to a complete consensus. The circle is also uninteresting if $p > nr$, since then there must be a pair of consecutive agents at distance greater than $r$ apart, so the configuration evolves as if it were living on the real line, with these two agents as the extreme opinions. Hence, when studying the circle, we may assume that $2r < p \leq nr$. 
Indeed, without loss of generality we can fix $r=1$ and have $p \in (2, \, n]$ as the only variable parameter. 
\par There are at least two possible motivations for considering the HK-dynamics on a circle:
\begin{enumerate}[(i)]
\item Firstly, it is possible to think of `real-life' situations where it makes more sense to consider opinions as lying on a circle instead of the real line. For example, the subject under debate may be the choice of a time of day or date on which to hold some event. 
\item Secondly, the circle is perhaps the simplest example of a space on which the HK-dynamics behave differently from in Euclidean space in a fundamental sense which we now describe.
\end{enumerate}
It is a well-known fact that, in Euclidean space, any configuration of opinions obeying the dynamics in (\ref{eq:update}) will freeze in finite time, that is, there will be some $T > 0$ such that $x_t (i) = x_T (i)$ for all $1 \leq i \leq n$ and $t \geq T$. In fact, it is known that the time taken to freeze is bounded above by a universal polynomial function of the number $n$ of agents. In one dimension, the best bound to date is $O(n^3)$, proven independently in \cite{BBCN} and \cite{MT}. In higher dimensions, the best published bound \cite{EB} is $O(n^8)$, but a recent preprint \cite{M} improves this to $O(n^4)$. Moreover, in a frozen configuration in $\mathbb{R}^k$, agents must either \emph{agree}, that is $x(i) = x(j)$, or be beyond each other's influence, that is $d(x(i), \, x(j)) > r$. 
\par On the other hand, if on the circle we place $n \geq p$ agents at angles $2 \pi j /n$, $0 \leq j < n$, then we have a frozen configuration in which no proper subset of the agents is isolated from its complement, but nor are any two agents in agreement. There also exist frozen configurations with these properties in which agents are not equally spaced. It is easy to see that any such configuration must include at least $4$ agents, and that the only example with $4$ agents, modulo translations, is that in Figure \ref{fig:stable} below. Furthermore, on the circle there are configurations which never freeze. For example, if we perturb the configuration in Figure \ref{fig:stable} by a sufficiently small amount then we will obtain one starting from which we never freeze, since the updated configurations will converge back to that in Figure \ref{fig:stable}, possibly translated, at an exponential rate. See Remark \ref{rem:4points} below for a proof of this last statement. 

\begin{figure}[h]
\begin{center}
\includegraphics[scale=1.5, trim = 70mm 0mm 70mm 0mm, clip,]{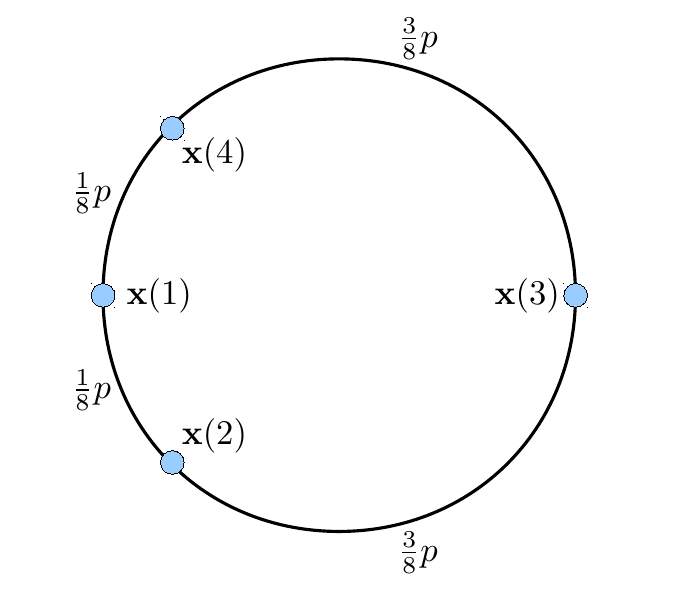}
\caption{A frozen configuration in which agents are not equally spaced. The perimeter $p$ should lie in the open interval $\left( 2, \, \frac{8}{3} \right)$.}
\label{fig:stable}
\end{center}
\end{figure}

Before proceeding, let us formally define the term \emph{convergence} for a finite set of agents obeying the HK-dynamics on a circle:
\begin{definition}\label{defi:converge}
Let $n \in \mathbb{N}$, $p \in \mathbb{R}_{> 2}$. Let $\bm{x}_0 = (x_0 (1), \dots, x_0 (n)) \in (\mathcal{C}_p)^n$, the latter equipped with the product topology. Suppose that, for all $t > 0$, the vectors $\bm{x}_t = (x_t (1),\dots, x_t (n))$ are defined inductively by the HK-update rule, with $r = 1$. 
If, for some $\bm{x}_{\infty} \in (\mathcal{C}_p)^n$, we have $\bm{x}_t \rightarrow \bm{x}_{\infty}$ in the product topology, then the sequence of configurations $\bm{x}_t$ is said to \emph{converge} to $\bm{x}_{\infty}$. It is common to abuse terminology slightly and say that the initial configuration $\bm{x}_0$ converges to $\bm{x}_{\infty}$.
\end{definition}  
In \cite{BHT} it is stated as an open problem for $\mathcal{C}_p$ 
whether any initial configuration will converge. The only other explicit mention of this problem that we could find is in \cite{Z} and, as far as we know, it has remained open up to now. Our first main result gives a positive answer:     
   
\begin{theorem}\label{thm:converge}
For any $n \in \mathbb{N}$, $p \in \mathbb{R}_{> 2}$ and $\bm{x}_0 \in (\mathcal{C}_p)^n$, the sequence $(\bm{x}_t)$ defined by the HK-update rule converges. 
\end{theorem}

The proof of this result will be given in the next section. It requires only moderate variations on ideas already presented in the literature when studying the HK-model in Euclidean space. Two important and well-known concepts are involved. Let $\bm{x} \in (\mathcal{C}_p)^n$. Firstly, we define the \emph{energy} of the configuration $\bm{x}$ as 
\begin{equation}\label{eq:energy}
E(\bm{x}) = \sum_{i=1}^{n} \sum_{j=1}^{n} \min \left\{ 1, \delta(x(i), \, x(j))^2 \right\}.
\end{equation}
Secondly, we define the \emph{influence graph} $G_{\bm{x}}$ corresponding to $\bm{x}$ to be the undirected 
graph whose vertices are $1,\, 2, \dots, \, n$ and where an edge is placed between $i$ and $j$ if and only if $\delta(x(i), \, x(j)) \leq 1$. Given an orientation of the circle, we can also define a \emph{directed influence graph} 
$\vec{G}_{\bm{x}}$ by choosing the oriented pair $(i, \, j)$ if and only if the shortest path from $i$ to $j$ follows the orientation.
\par The proof of Theorem \ref{thm:converge} consists of two main parts. Firstly, we show that an application of the HK-update rule always decreases the energy of a configuration, by an amount which is bounded below in terms of $\sum_{i=1}^{n} \delta(\bm{x}_{t}(i),\, \bm{x}_{t+1}(i))^2$. A similar result has been rediscovered several times over for Euclidean space, the earliest reference we could find is Theorem 2 of \cite{RMF}. Our proof for the circle is almost identical to that for Euclidean space given in Theorem 4.3 of \cite{BBCN}. Secondly, we show that if the digraphs $\vec{G}_{\bm{x}_{t-1}}$ and $\vec{G}_{\bm{x}_{t}}$ are different, that is, if at least one edge has been either added, deleted or changed orientation at time $t$, then the sum of the decreases in energy at times $t$ and $t+1$ is bounded below by $\frac{1}{9n^2}$. Since the initial energy is at most $n^2$, this implies that 
the digraph $\vec{G}_{\bm{x}_t}$ can change at most $18 n^4$ times and, in particular, must be fixed from some $t$ onwards. Once we know this, convergence will follow from standard linear algebra arguments. Our reasoning in the second part of the proof also closely resembles known arguments for Euclidean space, though the latter all seem to appeal at some point to the fact that the points of $\bm{x}_{t+1}$ all lie inside the convex hull of those in $\bm{x}_t$. There is no analogous fact for the circle, which gives the argument a new twist. The final proof of convergence, although it only appeals to basic linear algebra, also has no direct counterpart in Euclidean space since configurations always freeze eventually in the latter.
\\
\\
In Section 3, we will prove another result. Given that freezing occurs in Euclidean space after a time which is bounded by a universal polynomial function of the number $n$ of agents and given that, on the circle, the directed influence graph changes at most a polynomial-in-$n$ number of times, it is natural to ask whether freezing of the latter must also happen within a time which is bounded by a polynomial function of $n$ only. This turns out not to be the case. We shall prove that, for all sufficiently large $n$ and any $T \in \mathbb{N}$, there is a choice of perimeter $p$ and a configuration of $n$ agents whose influence digraph remains the same up to time $T$ but eventually changes. Indeed, we can choose the configuration so that $\vec{G}_{\bm{x}_t}$ is strongly connected all the way around the circle until the first change occurs, at which point a pair of adjacent agents become disconnected in the  underlying graph and the configuration subsequently collapses to a complete consensus after three additional time steps.  

Section 4 discusses some remaining open problems.


\section{Proof of Theorem \ref{thm:converge}}\label{sect:pfconverge}

The first step in the proof is to show that the HK-dynamics decrease the energy 
of a configuration as defined in (\ref{eq:energy}). The proof of the following result involves only minor modifications from that of Theorem 4.3 in \cite{BBCN}, but we will present it for the sake of completeness.

\begin{proposition}\label{prop:nrj}
If the sequence of configurations $(\bm{x}_t)$ obeys the HK-update rule then, for any $t\geq 0$, 
\begin{equation}\label{eq:minsk}
E(\bm{x}_{t+1}) \leq E(\bm{x}_t) - 4 \cdot \sum_{i=1}^{n} \delta (\bm{x}_{t+1} (i), \, \bm{x}_{t}(i))^{2}.
\end{equation}
\end{proposition}

\begin{proof} Fix an (anti-clockwise) orientation and an origin of $\mathcal{C}_p$, i.e.: a suitable bijection $\phi: \mathcal{C}_p \rightarrow [0, \, p)$. In what follows we shall often not distinguish in writing between a point $x \in \mathcal{C}_p$ and its image $\phi(x)$ - hopefully, it will always be clear form the context which is being referred to.                                              
\par For each $t \geq 0$, let $\bm{\Delta}_{t} = (\Delta_{t}(1),\dots,\, \Delta_{t}(n)) \in [-1, \, 1]^n$ be the vector such that
\begin{equation}\label{eq:Delta}
\Delta_{t}(i) \equiv \, x_{t+1}(i) - x_{t}(i) \, ({\hbox{mod $p$}}), \;\; i = 1,\dots,\, n.
\end{equation}
Note that (\ref{eq:minsk}) can be written in the form
\begin{equation}\label{eq:norminsk}
E(\bm{x}_{t+1}) \leq E(\bm{x}_t) - 4 \cdot ||\bm{\Delta}_t||^2,
\end{equation}
where $|| \cdot ||$ is the Euclidean norm on $\mathbb{R}^n$.   
\par For every $1 \leq i, \, j \leq n$ set
\begin{equation}\label{eq:r}
r_{t}(i, \, j) := \left\{ \begin{array}{lr} \delta(\bm{x}_{t}(i), \, \bm{x}_{t} (j)), & {\hbox{if the shortest path from $\bm{x}_{t}(i)$ to $\bm{x}_{t}(j)$ is
anti-clockwise,}} \\ -\delta(\bm{x}_{t}(i), \, \bm{x}_{t}(j)), & {\hbox{otherwise.}} \end{array} \right.
\end{equation}
We write $i \sim_{t} j$ if $\delta(\bm{x}_{t}(i), \, \bm{x}_{t}(j)) \leq 1$. Define functions $\mathcal{E}_t: \mathbb{R}^n \times \mathbb{R}^n \rightarrow \mathbb{R}$ as follows:
\begin{equation}\label{eq:scripte}
\mathcal{E}_t (\bm{u}, \, \bm{v}) := \sum_{i=1}^{n} \sum_{j=1}^{n} \mathcal{E}_t (\bm{u}, \, \bm{v}, \, i, \, j),
\end{equation}
where, if $\bm{u} = (u_1, \dots, \, u_n)$ and $\bm{v} = (v_1, \dots, \, v_n)$, then
\begin{equation}\label{eq:scripteij}
\mathcal{E}_{t}(\bm{u}, \, \bm{v}, \, i, \, j) := \left\{ \begin{array}{lr}
(u_i - v_j - r_{t}(i, \, j))^2, & {\hbox{if $i \neq j$ and $i \sim_{t} j$}}, \\
1, & {\hbox{if $i \neq j$ and $i \not\sim_{t} j$}}, \\
(u_i + v_j)^2, & {\hbox{if $i = j$}}. \end{array} \right.
\end{equation}
Clearly the functions $\mathcal{E}_t$ are symmetric, i.e.: $\mathcal{E}_t (\bm{u}, \, \bm{v}) = \mathcal{E}_t (\bm{v}, \, \bm{u})$, and 
\begin{equation}\label{eq:ee}
E(\bm{x}_t) = \mathcal{E}_{t} (\bm{0}, \, \bm{0}).
\end{equation}
It is easy to see that the function $\bm{u} \mapsto \mathcal{E}_t (\bm{u}, \, \bm{0})$ is strictly convex on $\mathbb{R}^n$ and attains its global minimum at 
$\bm{u} = \bm{\Delta}_t$. Hence if we define $f_{t}: \mathbb{R}^2 \rightarrow \mathbb{R}$ by $f_{t}(z,\, w) = \mathcal{E}_t (z\bm{\Delta}_t, \, w \bm{\Delta}_t)$ then
\begin{equation}\label{eq:argmin}
1 = {\hbox{argmin}}_{z} \, f_{t}(z,\, 0) = {\hbox{argmin}}_{w} \, f_{t}(0, \, w).
\end{equation}
It is also clear that $f_t$ is a convex second-degree polynomial in $z$ and $w$. Thus it must have the form 
\begin{equation}\label{eq:f}
f_{t}(z,\,w) = A(z-1)^2 + A(w-1)^2 + Bzw + C, \;\;\; {\hbox{for some $A \geq 0$, 
$-2A \leq B \leq 2A$}}.
\end{equation}
In particular, $f(1,\,1) \leq f(0,\,0)$ which means that
\begin{equation}\label{eq:ineq1}
\mathcal{E}_{t}(\bm{\Delta}_t, \, \bm{\Delta}_t) \leq \mathcal{E}_t (\bm{0}, \, \bm{0}).
\end{equation}
Secondly, we claim that 
\begin{equation}\label{eq:ineq2}
\mathcal{E}_{t}(\bm{\Delta}_t, \, \bm{\Delta}_t) - \mathcal{E}_{t+1}(\bm{0}, \, \bm{0}) \geq 4 \cdot ||\bm{\Delta}_t||^2.
\end{equation}
Note that (\ref{eq:ineq2}), (\ref{eq:ineq1}) and (\ref{eq:ee}) would together imply (\ref{eq:norminsk}). To prove (\ref{eq:ineq2}) we simply note that, on the one hand, if $1 \leq i \leq n$ then
\begin{equation}\label{eq:ineq3}
\mathcal{E}_{t}(\bm{\Delta}_t, \, \bm{\Delta}_t, \, i, \, i) = (2 \Delta_{t}(i))^2, \;\;\;\;\;\;\;\; \mathcal{E}_{t+1}(\bm{0}, \, \bm{0}, \, i, \, i) = 0,
\end{equation}
while, if $i \neq j$, then it is easily verified that 
\begin{equation}\label{eq:eq1}
|r_{t+1}(i, \, j)| \leq |r_{t}(i, \, j) - (\Delta_{t}(i) - \Delta_{t}(j))|
\end{equation}
and hence 
\begin{equation}\label{eq:ineq4}
\mathcal{E}_t (\bm{\Delta}_t, \, \bm{\Delta}_t, \, i, \, j) \geq \mathcal{E}_{t+1}(\bm{0}, \, \bm{0}, \, i, \, j), \;\;\;\; i \neq j.
\end{equation}
Together (\ref{eq:ineq4}) and (\ref{eq:ineq3}) imply (\ref{eq:ineq2}), so the proof of the proposition is complete.
\end{proof}

The second step in the proof of Theorem \ref{thm:converge} is to establish a lower bound on the decrease in energy resulting from any change in the directed influence graph.

\begin{proposition}\label{prop:graph}
Let $t > 0$ and suppose that $\vec{G}_{\bm{x}_{t-1}} \neq \vec{G}_{\bm{x}_{t}}$. Then, with notation as in (\ref{eq:Delta}), 
\begin{equation}\label{eq:change}
||\bm{\Delta}_t||^2 + ||\bm{\Delta}_{t+1}||^2 \geq \frac{1}{9n^2}.
\end{equation}
\end{proposition}

\begin{proof}
We again suppose that an anti-clockwise orientation of the circle has been fixed. The terms left and right will be used synonymously with clockwise and anti-clockwise respectively. If $\vec{G}_{\bm{x}_{t-1}} \neq \vec{G}_{\bm{x}_t}$ then a priori at least one of the following four things must have happened at time $t$:
\begin{enumerate}[(i)]
\item some agent, $i_1$ say, has a new neighbour $i_2$ to his right. In other words, the oriented edge $(i_1, \, i_2)$ is in $\vec{G}_{\bm{x}_t}$ but not in $\vec{G}_{\bm{x}_{t-1}}$. Note that it could still be the case (if $p \leq 3$) that $(i_2, \, i_1) \in \vec{G}_{\bm{x}_{t-1}}$, it will not matter for the argument we present below.
\item some agent $i_1$ has a new neighbour $i_2$ to his left.
\item some agent $i_1$ has a lost a neighbour $i_2$ on his right.
\item some agent $i_1$ has lost a neighbour $i_2$ on his left.
\end{enumerate}
We suppose that (i) occurs - analogous arguments can be given in each of the other three cases. Pick an oriented edge $(i_1, \, i_2)$ which is added at time $t$. If either of these agents moved a distance of at least $\frac{1}{3n}$ at step $t$, then $||\bm{\Delta}_{t}||^2 \geq \frac{1}{9n^2}$ and we are done. Otherwise, the distance between these two agents at time $t$ is still at least $1 - \frac{2}{3n}$. Now agent $i_1$ will draw agent $i_2$ to the left (i.e.: clockwise) at step $t+1$. If agent $i_2$ were to move a distance of at least $\frac{1}{3n}$ to the left at step $t+1$ then we would be done again. But since $i_2$ gained a neighbour to his left at time $t$ and since the ordering of agents around the circle is always preserved by the HK-dynamics, he cannot have lost any neighbours to his left at time $t$. Hence, there are only two possible ways he can be prevented from moving at least $\frac{1}{3n}$ to his left at time $t+1$:
\begin{enumerate}[(a)]
\item there is some $i_3$ such that the oriented edge $(i_2, \, i_3)$ is also added to the digraph at time $t$, or 
\item the neighbours of $i_2$ at time $t-1$, including himself, move on average sufficiently far to the right at time $t$ to compensate for the appearance of $i_1$. Let $M$ be the total net movement to the right of these neighbours at time $t$. Then, since $i_2$ has certainly no more than $n$ neighbours at time $t$ and since $i_1$ is such a neighbour, his movement to the left at time $t+1$ is bounded below by $\frac{1}{n} \left( 1 - \frac{2}{3n} - M \right)$. Thus $M \geq \frac{2}{3} \left( 1 - \frac{1}{n} \right)$. Since $i_2$ has at most $n-1$ neighbours at time $t-1$, by the Cauchy-Schwarz inequality the sum of their contributions to $||\Delta_{t}||^2$ is at least $\frac{M^2}{n-1}$, which is greater than $\frac{1}{9n^2}$ for all $n \geq 3$. 
\end{enumerate}
Thus we are done again unless (a) occurs. Now we can apply the same analysis to the pair $(i_2, \, i_3)$ and so on. For the proposition to fail, we would have to have an infinite, anti-clockwise sequence of agents $i_1, \, i_2, \dots$ such that each oriented edge $(i_k, \, i_{k+1})$ appeared at time $t$. But this would mean that, for every $k$, agent $i_{k+1}$ moved further to the left than $i_k$ at time $t$. Since there are only finitely many agents, this is impossible. 
\end{proof}
   
\begin{proof} \emph{of Theorem \ref{thm:converge}}. The energy of a configuration of $n$ agents is certainly no more than $n^2$. Hence, if the sequence of configurations $(\bm{x}_t)$ obeys the HK-update rule, it follows from Propositions \ref{prop:nrj} and \ref{prop:graph} that there are at most $18 n^4$ values of $t$ for which $\vec{G}_{\bm{x}_t} \neq \vec{G}_{\bm{x}_{t+1}}$. Thus it suffices to prove Theorem \ref{thm:converge} under the assumption that the directed influence graph $\vec{G}$ is the same for all time. Furthermore, we may assume 
that the underlying graph $G$ is connected, as otherwise the dynamics are the same as for a configuration on the real line, in which case we know it will freeze in finite time. 
\par Once again, let us fix a choice of a suitable bijection $\phi: \mathcal{C}_p \rightarrow [0,\, p)$. The map $\phi$ induces an ordering of the components of any $\bm{x} \in \mathcal{C}_{p}^{n}$. Let us assume that, for our initial configuration $\bm{x}_0 = (x_0 (1),\dots,\,x_0 (n))$, we have $\phi(x(1)) \leq \phi(x(2)) \leq \dots \leq \phi(x(n))$. Then, since the HK-dynamics will never cause agents to cross, for all $t$ the agents $1, \, 2, \dots, \, n$ will retain an anti-clockwise ordering, even though their ordering with respect to $\phi$ may be cyclically shifted. 
\par Let $\mathcal{P}^{n}$ denote the set of probability vectors in $\mathbb{R}^n$, i.e.: $\mathcal{P}^n = \{(v_1,\dots,\,v_n)^T \in \mathbb{R}^n: v_i\geq 0, \, \sum_{i} v_i = 1\}$. Then there is a natural map $\bm{x}_t \mapsto \bm{x}^{*}_{t} \in \mathcal{P}^n$ defined by 
\begin{equation}\label{eq:distform}
\bm{x}^{*}_{t} = (x^{*}_{t}(1),\dots,\, x^{*}_{t}(n)) \;\; {\hbox{where}} \;\; p \cdot x^{*}_{t}(i) = \left\{ \begin{array}{lr} \delta(x_{t}(i+1), \, x_{t}(i)), & {\hbox{if $1 \leq i \leq n-1$}}, \\  \delta(x_{t}(1)), \, x_{t}(n)), & {\hbox{if $i = n$}}. \end{array} \right.
\end{equation}
Since the directed influence graph is fixed, there is a fixed $n \times n$ real matrix $A$ such that the dynamics satisfy{\footnote{Note that it is crucial that the directed influence graph be fixed for the same to be true of the transition matrix, it would not suffice in general for the underlying graph to be fixed.}} 
\begin{equation}\label{eq:distdyn}
\bm{x}^{*}_{t+1} = A \, \bm{x}^{*}_{t}, \;\;\;\; {\hbox{for all $t \geq 0$}}.
\end{equation}

By Proposition \ref{prop:nrj}, we know that the differences $\bm{x}^{*}_{t+1} - \bm{x}^{*}_{t}$ tend to $\bm{0}$ as $t\rightarrow\infty$. Hence, letting $\bm{v}:=(A-I_n) \bm{x}^*_0$, we have
\begin{equation}\label{eq:xdiff}
\bm{x}^*_{t+1}-\bm{x}^*_t = A^t \bm{v} \rightarrow \bm{0} \text{ as }t\rightarrow\infty.
\end{equation}
Hence, $\bm{v}$ must be identically zero in those coordinates corresponding to each block in the Jordan normal form of $A$ whose eigenvalue is at least one in absolute value. It follows that there exists a $c>0$, only depending on $A$, such that $\|\bm{x}^*_{t+1}-\bm{x}^*_t\| = O(e^{-ct})$. Thus $\bm{x}^*_t$ must converge to a limit, call it $\bm{x}^*_\infty$, such that
\begin{equation}\label{eq:xinfty}
||\bm{x}^*_t-x^*_\infty||=O(e^{-ct}).
\end{equation}
By the definition of energy in (\ref{eq:energy}), it follows that there is some $E_{\infty} \in \mathbb{R}_{\geq 0}$ such that $E(\bm{x}_t) = E_{\infty} + O(e^{-\Omega(t)})$ and hence in turn, by Proposition \ref{prop:nrj}, that $||\bm{\Delta}_t||^2 = O(e^{-\Omega(t)})$. This establishes convergence of the sequence $(\bm{x}_t)$, in the sense of Definition 1.1. 
\end{proof}

\begin{rem}\label{rem:irredA}
An alternative argument for obtaining (\ref{eq:xinfty}) is as follows. It is easy to see that the matrix $A$ in (\ref{eq:distdyn}) has non-negative entries. We claim that $A$ is column-stochastic. To prove this, it suffices to prove that 
\begin{equation}\label{eq:ones}
\bm{1}^T A \bm{v} = \bm{1}^T \bm{v} \;\; {\hbox{for all $\bm{v} \in \mathbb{R}^n$, where $\bm{1} = (1,\dots,\, 1)^T$}}. 
\end{equation}
By definition, (\ref{eq:ones}) holds for $\bm{v} = \bm{x}^{*}_t$ and any $t$. Now let $\bm{v}$ be any probability vector. The transformation $\bm{v} \mapsto A\bm{v}$ describes an updating of a configuration of $n$ agents, possibly involving a cyclic shift of the agents' indices. This may not coincide with the HK-update rule for this configuration as there is no guarantee that agents interact precisely with those within unit distance of themselves. However, exactly the same pairs of agents interact as would be the case if the configuration were some $\bm{x}_t$. In particular, this means that the anti-clockwise ordering of agents will be preserved (no agents will cross), which implies that $A\bm{v}$ will also be a probability vector. Thus (\ref{eq:ones}) holds for any probability vector, and hence for any $\bm{v} \in \mathbb{R}^n$. 
\par Now $A$ may not be regular, since there may be pairs of agents that have exactly the same neighbours in $\vec{G}$ at $t = 0$. However, any such pair will immediately reach consensus at $t=1$. If we start at $t=1$ and replace the vector $\bm{x}^{*}_{t}$ of (\ref{eq:distform}) by the (possibly shorter) vector of non-zero distances, then it is easy to see that we have an analogue of (\ref{eq:distdyn}) in which the transition matrix is now regular. Thus we will have exponential convergence to the steady-state, which yields (\ref{eq:xinfty}).
\end{rem}
\begin{rem}\label{rem:4points}
With notation as above we have, for the configuration $\bm{x} \in (\mathcal{C}_p)^4$, $p \in \left( 2, \, \frac{8}{3} \right)$, of Figure 1 that 
\begin{equation}\label{eq:4ex}
\bm{x}^{*} = \left[ \begin{array}{c} 1/8 \\ 3/8 \\ 3/8 \\ 1/8 \end{array} \right], \;\;\;\; \bm{x}^{*} = A\bm{x}^{*} \;\; {\hbox{where}} \;\; A = \left[ \begin{array}{cccc} 1/6 & 1/4 & 0 & 1/12 \\ 1/2 & 5/12 & 1/3 & 1/4 \\ 1/4 & 1/3 & 5/12 & 1/2 \\ 1/12 & 0 & 1/4 & 1/6 \end{array} \right].
\end{equation}
The second eigenvalue of $A$ is $\lambda_2 = \frac{1}{3}$ with eigenvector 
$\bm{v}_2 = (1, \, 1, \, -1, \, -1)^T$. If we choose $\bm{x}_0$ such that $\bm{x}^{*}_{0} = \bm{x}^{*} + \varepsilon \bm{v}_2$ then, for sufficiently small but non-zero 
$\varepsilon$ (how small depends on how close $p$ is to $2$ or to $8/3$), the HK-updates will satisfy (\ref{eq:distdyn}) and the configurations $\bm{x}_t$ will converge to some translation of $\bm{x}_0$ without ever actually reaching it.  
\end{rem}

\section{Influence graphs which take arbitrarily long to freeze}\label{sect:longfreeze}

In this section we will prove that the time taken for the (directed or undirected) influence graph to freeze is unbounded as a function of $n$. To simplify some notation, we will assume the perimeter $p$ to be fixed and instead allow $r$ to vary in the interval $\left( \frac{p}{n}, \, \frac{p}{2} \right)$.
 
\begin{theorem}\label{thm:longfreeze}
For any sufficiently large $n$ there is an initial configuration with $n$ agents such that the time it takes for the influence graph to freeze can be made arbitrarily large by picking $r$ appropriately. Hence, the time until the influence graph freezes is unbounded solely as a function of $n$.
\end{theorem}
\begin{proof}
Pick positive integers $m_1, m_2$ such that $m_1 < m_2 < 2 m_1$ and $n=3m_1 + 2 m_2$. This is possible provided $n$ is sufficiently large{\footnote{Indeed, there is a solution for any $n \geq 12$ except if $n \in \{13, \, 14, \, 15, \, 16, \, 18, \, 20, \, 21, \, 23, \, 25, \, 28, \, 30, \, 35\}$.}}.

Consider a state consisting of five clusters{\footnote{By a \emph{cluster} we mean a set of agents in agreement.}} of sizes $m_1, \, m_2, \, m_1, \, m_1, \, m_2$ as illustrated in Figure \ref{fig:slowconv}. If $d_2 = \left(\frac{m_2}{m_1}\right) d_1$ then such a configuration will be stable for any choice of $r \in \left( d_2, \, 2d_1 \right)$. Let $\bm{x}$ denote such a stable state and for $y_0(1), \, y_0(2)$ sufficiently small positive numbers, let $\bm{x}_0$ denote the perturbation of $\bm{x}$ as indicated in the figure. 

\begin{figure}[h]
\begin{center}
\includegraphics[scale=1, trim = 50mm 0mm 50mm 0mm, clip,]{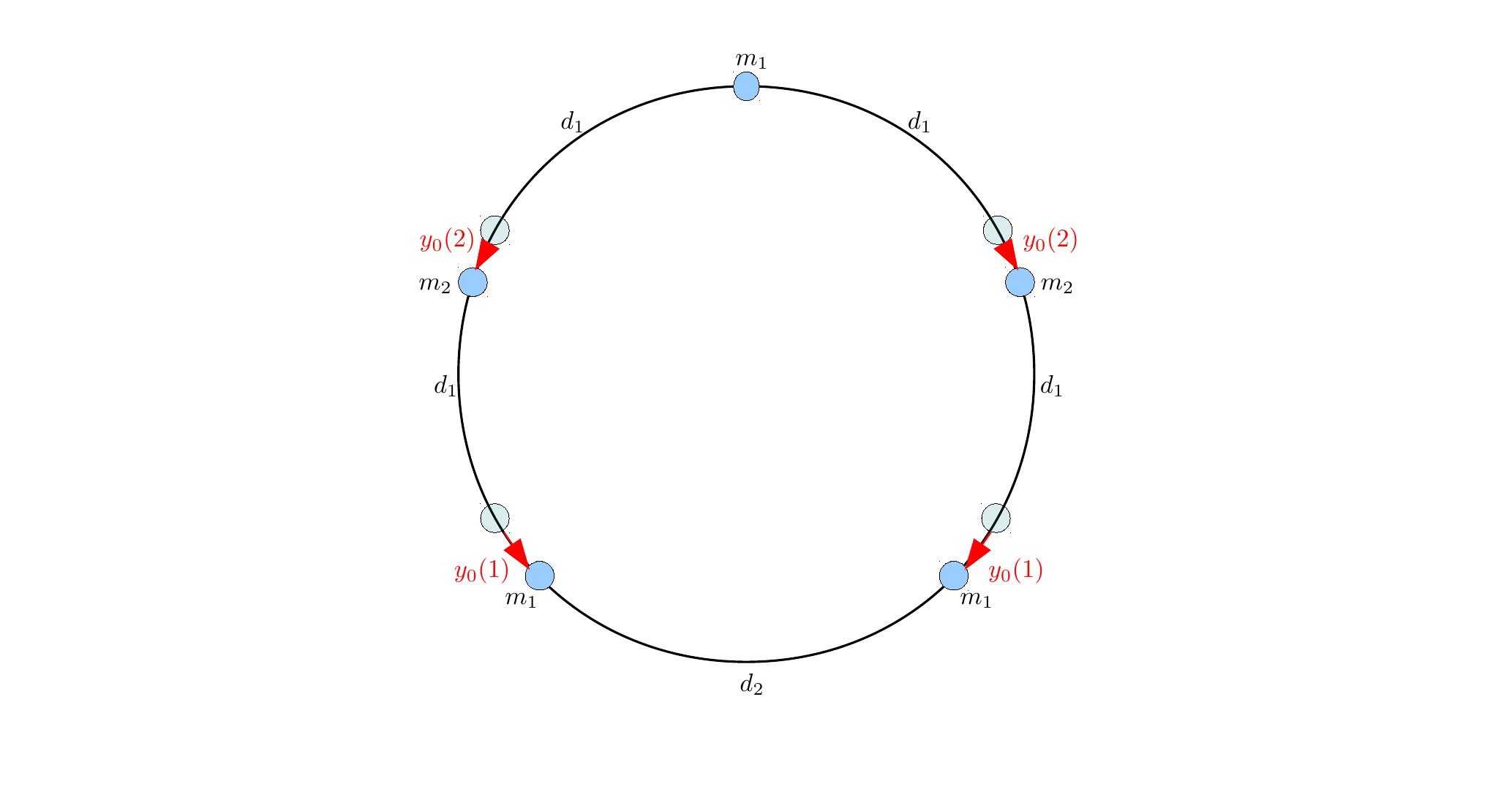}
\caption{A state consisting of $n$ agents in $5$ clusters. We assume that every cluster sees one other cluster on each side. The inter-cluster distances $d_1, \, d_2$ refer to the stable state. The arrows indicate the perturbation from the stable state we consider.}
\label{fig:slowconv}
\end{center}
\end{figure}


It is easy to see that, as long as every cluster in $\bm{x}_t$ sees exactly one other cluster on either side, the state at time $t+1$ is the perturbation of $\bm{x}$ by $y_{t+1}(1)$ and $y_{t+1}(2)$, defined recursively by
\begin{align}
y_{t+1}(1) &= \frac{m_2 y_t(2)}{2m_1+m_2}\\
y_{t+1}(2) &= \frac{m_1 y_t(1) + m_2 y_t(2)}{2m_1+m_2}.
\end{align}
Note that these sequences are positive, bounded by $\max\left\{y_0(1), \, y_0(2)\right\}$ and tending to $0$ as $t\rightarrow\infty$. In particular, this means that the distance between the adjacent clusters of equal size in $\bm{x}_t$ is strictly less than, but tending to $d_2$. We assume $y_0(1)$ and $y_0(2)$ are sufficiently small such that the distance between any other pairs of clusters remains bounded away from $d_2$.

To complete the proof of the theorem, consider the HK-dynamics with initial state $\bm{x}_0$ as described above for $r$ slightly smaller than $d_2$. In this case, the influence graph will be constant until the distance between the clusters of equal size is greater than $r$, at which time the model behaves as five clusters on $\mathbb{R}$. The time before this happens can be made arbitrarily long by choosing $r$ sufficiently close to $d_2$.
\end{proof}

\begin{rem}\label{rem:consensus}
Suppose we run the HK-dynamics on the real line, starting from five equally spaced clusters of sizes $m_1, \, m_2, \, m_1, \, m_2$ and $m_1$, reading from left to right, with inter-cluster spacing $d_1$ and $r \in \left(d_1, \, d_2 \right)$, where $d_2 = \left( \frac{m_2}{m_1} \right) d_1$ and $m_1 < m_2 < 2m_1$. It is easy to see that any such configuration will eventually reach a consensus. Hence, on the circle, the same will be true of the configuration in Theorem \ref{thm:longfreeze} for $r$ sufficiently close to $d_2$. We now show that the various parameters can be chosen so that, on the real line, consensus occurs after three time steps and hence, on the circle, three time steps after the graph disconnects. 
\par First consider the real line again. At $t=1$, the two extreme clusters will move inwards a distance $\left( \frac{m_2}{m_1 + m_2} \right) d_1$ while the other three clusters will remain fixed. Hence if $\left(1 + \frac{m_1}{m_1 + m_2} \right) d_1 < r < d_2$, which can be satisfied provided $m_2 > \sqrt{2} \, m_1$, then the middle cluster of size $m_1$ will be visible to all the others at $t=1$. Hence at $t=2$ we will have clusters of size $m_1 + m_2$ on either side of the middle cluster, and each at a distance of $\left( \frac{m_{1}^{2} + m_{2}^{2} + m_1 m_2}{m_{2} (2m_1 + m_2)} \right) \, d_1$ from it. Thus if 
\begin{equation}\label{eq:collapse}
\frac{r}{d_1} \in \left( \frac{2(m_{1}^{2} + m_{2}^{2} + m_1 m_2)}{m_{2} (2m_1 + m_2)}, \, \frac{m_2}{m_1} \right),
\end{equation}
then the configuration will collapse to a complete consensus at $t=3$. One may check that the open interval in (\ref{eq:collapse}) is non-empty if $\xi_0 < \frac{m_2}{m_1} < 2$, where $\xi_0 \approx 1.7693 ...$ is the unique real root of the equation $\xi^{3} - 2\xi - 2 = 0$. It follows that also on a circle of fixed perimeter, for all $n$ sufficiently large and all $T \in \mathbb{N}$, there is a choice of $r$ such that the configuration in Theorem \ref{thm:longfreeze} will remain connected up to time $T$ and, once it disconnects, will collapse to consensus after three further time steps. 
\end{rem}

\section{Open Problems}\label{sect:open}

It is natural to ask whether the methods of this paper can be adapted to prove convergence of the HK-dynamics in even more general spaces than the circle. The case of higher-dimensional tori can be motivated in the usual manner, by imagining that there is more than one issue on which agents have opinions. Generalisation to even more abstract spaces would at present seem to be a purely mathematical curiosity. One of the main challenges in such work would be to give a mathematically rigorous formulation of the problem. 
\par On the circle, there remain some questions of possible interest. One could investigate how the time to freezing of the influence digraph behaves for fixed $p$ and $r$. This would seem to depend very subtly on either parameter. Another question that occurred to us but which we cannot presently resolve is whether the sequence of influence digraphs $\vec{G}_{\bm{x}_t}$ resulting from a given initial configuration must always be acyclic - in other words, if $\vec{G}_{\bm{x}_s} = \vec{G}_{\bm{x}_t}$ for some $s < t$, must it be the case that $\vec{G}_{\bm{x}_u} = \vec{G}_{\bm{x}_v}$ for all $s \leq u, \, v \leq t$ ?

\section{Acknowledgements}
We thank Bernadette Charron-Bost for pointing out a small error in an earlier version of the manuscript, in which we incorrectly stated that the matrix $A$ in (\ref{eq:distdyn}) was regular. Remark \ref{rem:irredA} presents the corrected assertions.


\begin{thebibliography}{99} 

\bibitem{BBCN} A. Bhattacharya, M. Braverman, B. Chazelle and H. L. Nguyen, \emph{On the convergence of the Hegselmann-Krause system}, Proceedings of the 4th Innovations in Theoretical Computer Science conference (ICTS 2013), Berkeley CA, January 2013.


\bibitem{BHT} V. D. Blondel, J. M. Hendrickx and J. N. Tsitsiklis, \emph{On Krause's multi-agent consensus model with state-dependent connectivity}, IEEE Trans. Automat. Control {54} (2009), No. 11, 2586--2597.




\bibitem{EB} S. R. Etesami and T. Basar, \emph{Game-theoretic analysis of the Hegselmann-Krause model for opinion dynamics}, to appear in IEEE Trans. Autom. Control. Available online at \texttt{http://ieeexplore.ieee.org/stamp/stamp.jsp?arnumber=7024142}

\bibitem{HK} R. Hegselmann and U. Krause, \emph{Opinion dynamics and bounded confidence: models, analysis and simulations}, Journal of Artificial Societies and Social Simulation {5} (2002), No. 3. Fulltext at \texttt{http://jasss.soc.surrey.ac.uk/5/3/2/2.pdf}  

\bibitem{MBCF} S. Martinez, F. Bullo, J. Cortes and E. Frazzoli, \emph{On Synchronous Robotic Networks - Part II: Time Complexity of Rendezvous and Deployment Algorithms}, IEEE Trans. Automat. Control {52} (2007), No. 12, 2214--2226.

\bibitem{M} A. Martinsson, \emph{An improved energy argument for the Hegselmann-Krause model}. Preprint at \texttt{http://arxiv.org/pdf/1501.02183.pdf}

\bibitem{MT} S. Mohajer and B. Touri, \emph{On convergence rate of scalar Hegselmann-Krause dynamics}. Preprint at \texttt{http://arxiv.org/pdf/1211.4189v1.pdf}

\bibitem{RMF} M. Roozbehani, A. Megretski and E. Frazzoli, \emph{Lyapunov analysis of quadratically symmetric neighborhood consensus algorithms}, Proceedings of the 47th IEEE Conference on Decision and Control (CDC 2008), pp. 2252--2257.

\bibitem{Z} M. Zahri, \emph{Condensing on metric spaces: modeling, analysis and simulation}, Ph.D. thesis, Goethe University Frankfurt (2009). 

\bibitem{ZS} S. Zhang and Z. Sun, \emph{Consensus analysis of a $2$-dimensional KH model with various confidence level sets}, Proceedings of the Second International Conference on Intelligent Computation Technology and Automation (ICICTA 2009), IEEE Computer Society, Washington DC (2009). 






\end{thebibliography}
\end{document}